\documentclass{ifacconf}

\usepackage{color}
\usepackage{booktabs} 
  \usepackage{amsmath}
\usepackage{amssymb}
\usepackage{amsbsy}
\usepackage{algorithm}               
\usepackage{algpseudocode}
\usepackage{graphicx}  
\usepackage{psfrag}
\usepackage{diagbox}
\usepackage{bm}
\usepackage{bbm}                                                       
\allowdisplaybreaks

\usepackage{bm}
\usepackage{bbm}
    
\RequirePackage{epsfig}
\RequirePackage{psfrag}
\RequirePackage{pstool}
\RequirePackage{epstopdf}
\RequirePackage{graphicx}
\usepackage{bbm}
\usepackage{subfigure}
\RequirePackage{mathrsfs}

\RequirePackage{amsmath}
\RequirePackage{amssymb}

\RequirePackage{amsbsy}
\RequirePackage{amsfonts}
\RequirePackage{mathrsfs}
\usepackage{algorithm}   
\usepackage{algpseudocode}

\usepackage{psfrag} 
\usepackage{amsmath}

\newtheorem{definition}{Definition}
\newtheorem{lemma}{Lemma}
\newtheorem{theorem}{Theorem}
\newtheorem{remark}{Remark}

\newtheorem{assumption}{Assumption}

\usepackage{algorithm}   
\usepackage{algpseudocode}

\usepackage{enumitem}
\usepackage{tabularx,booktabs,ragged2e}

\usepackage{graphicx}      
\usepackage{natbib}        

\begin{document}
\begin{frontmatter}

\title{Contraction Analysis of Filippov Solutions in Multi-Modal Piecewise Smooth Systems} 

\author[First]{Zonglin Liu} 
\author[First]{Kyra Borchhardt} 
\author[First]{Olaf Stursberg} 

\address[First]{Control and System Theory, Dept. of Electrical Engineering and Computer Science, University of Kassel, Germany.\\
	Email: z.Liu@uni-kassel.de,  uk089537@student.uni-kassel.de, stursberg@uni-kassel.de.}

\begin{abstract}
This paper provides conditions to ensure contractive behavior of Filippov solutions generated by multi-modal piecewise smooth (PWS) systems. These conditions are instrumental in analyzing the asymptotic behavior of PWS systems, such as convergence towards an equilibrium point or a limit cycle.
 The work is motivated by a known principle for contraction analysis of bimodal PWS systems which ensures that the flow dynamics of each mode and the sliding dynamics on the switching manifold are contracting. This approach is extended first to PWS systems  with multiple non-intersecting switching manifolds in  $\mathbb{R}^n$, and then to two intersecting switching manifolds in  $\mathbb{R}^2$. Numerical examples are provided to validate  the theoretical findings, along with a discussion on extensions to more general intersecting switching manifolds in  $\mathbb{R}^n$.
\end{abstract}
	

\end{frontmatter}
\section{Introduction}\label{sec:intro}

A contractive dynamical system is characterized by the property that all trajectories converge to one another, regardless of the  initial states of the system. Contraction analysis can therefore be used to study various asymptotic properties, such as the stability of equilibrium points or periodic orbits, see \cite{sontag2010contractive}. In recent years, this property has also been utilized to design systems that share a desired stable limit cycle  by \cite{hanke2025approximation}, to develop coupling rules that achieve asymptotic synchronization by \cite{aminzare2020cluster}, and to synthesize tracking controllers and state observers by \cite{miljkovic2025reference}.

For systems described by smooth differential equations, a bound on the matrix measure of the Jacobian can be used to prove the contractivity of classical solutions \cite{lohmiller1998contraction, aminzare2014contraction}. Regarding non-smooth dynamics, the work in \cite{pavlov2007convergence} proposed a contraction condition for a class of piecewise affine systems with continuous vector fields on switching boundaries. Accounting for more general piecewise smooth (PWS) systems with non-continuous vector fields on  switching manifolds, the work by \cite{di2013incremental} established conditions for contracting Filippov solutions, including sliding motion which constitutes the main difference to standard solutions of smooth dynamics of planar PWS systems in $\mathbb{R}^2$. This result was later extended to bimodal PWS systems in $\mathbb{R}^n$ by \cite{di2014incremental}. The latter two papers ensure contraction of the sliding dynamics by combining the contracting local dynamics on both sides of the switching boundary. Facing the restriction that the sliding dynamics must be defined everywhere rather than only on the switching boundaries in those papers, the method of regularization (which is common to smoothen  discontinuous vector fields on switching boundaries in different contexts) was adopted in \cite{fiore2016contraction} to establish a contraction condition for Filippov solutions of bimodal PWS systems  in  $\mathbb{R}^n$.
For a class of multi-modal  PWS systems  in  $\mathbb{R}^n$ or for piecewise smooth continuous systems (PWSC) as defined in \cite{bernardo2008piecewise}, a condition ensuring contracting Caratheodory solutions is provided by \cite{di2014contraction} and applied to guarantee synchronization of a network of such systems.
For a class of hybrid dynamic systems, the work in \cite{tang2014transverse, burden2018contraction} proposed contractivity conditions for solutions in the case that the system states change impulsively through reset functions upon mode switching.

The literature review shows that conditions for contractivity of Filippov solutions in multi-modal PWS systems in $\mathbb{R}^n$ are not yet been established. A possible reason may be the non-uniqueness of the sliding vector at intersections of the switching boundaries (see \cite{kaklamanos2019regularization}), which complicates the evaluation of whether the sliding motion is contracting. After introducing some fundamentals of PWS and PWSC systems in  Sec.~2, the identified gap is addressed in Sec. 3 by proposing a scheme for contraction analysis of a class of multi-modal PWS systems in $\mathbb{R}^n$ with parallel switching boundaries, thus directly extending the case of bimodal systems in \cite{fiore2016contraction}. The focus is shifted in Sec.~4 to  a class of PWS systems in $\mathbb{R}^2$ with two intersecting switching boundaries in order  to provide initial insight into handling intersections in contraction analysis.  This result is afterwards generalized to PWS systems in  $\mathbb{R}^n$ with arbitrarily many intersecting switching boundaries. Sec.~5 presents two numerical examples of intersecting and nonintersecting switching manifolds, and the paper concludes in Sec.~6 with an outlook on future research.

\section{Contracting PWSC Systems and Filippov Solutions }\label{subsec:ContractivePWSC}

The class of piecewise smooth (PWS) systems  considered in this paper consists of a finite set  of ordinary differential equations:
\begin{align}\label{eq:generalPWSnew}
\dot{x}=F(x)=\left\{\begin{array}{cc}
  x \in \mathcal{S}_i \subset \mathbb{R}^n, i \in \mathcal{N}=\{1,\ldots, N\}:     \\
 \dot{x}=f_i(x)   
\end{array}\right\}   
\end{align}
in which the smooth vector fields $f_i(x)$ are defined on disjoint open sets  $\mathcal{S}_i$ and are continuously extendable to the closure (or boundary) of   $\mathcal{S}_i$. The union of   $\mathcal{S}_i$ and  its closure is denoted by $\bar{\mathcal{S}}_i$. The disjoint sets are separated by a set $\Omega$ of \textit{switching manifolds} $\Sigma_{i,j}= \bar{\mathcal{S}}_i \cap  \bar{\mathcal{S}}_j$
with codimension one. The union of $\Omega$ and all  $\mathcal{S}_i$,  $i \in \mathcal{N}$, covers the whole state space $\mathbb{R}^n$.
\begin{definition}\label{def:PWSC} (\cite{di2014contraction})
The PWS system \eqref{eq:generalPWSnew} is called \textit{piecewise smooth continuous} (PWSC) if:
\begin{enumerate}
\item $F(x)$ is continuous for all $x \in \mathbb{R}^n$,
\item and the $f_i(x)$,  $i \in \mathcal{N}$, are continuously differentiable for all $x \in  \mathcal{S}_i$ and the Jacobians $\frac{\partial f_i}{\partial x}(x)$ can be continuously extended to the closure of   $\mathcal{S}_i$.
\end{enumerate}
\end{definition}
Clearly, a PWSC system represents a special case of a PWS system in which the vector field on the switching manifold remains continuous.

For a positive definite matrix  $Q  \in \mathbb{R}^{n \times n} \succ 0$, the matrix norm $\mu_{Q}(A)$ for  $A \in \mathbb{R}^{n \times n}$  is defined by:
\begin{align}\label{eq:normproperty}
\mu_{Q}(A)=  \lambda_{max}(\frac{1}{2}(QAQ^{-1}+Q^{-1}A^TQ))
\end{align}
in which $\lambda_{max}$ denotes the largest eigenvalue.
It is known that $\mu_{Q}(cA) = c \mu_{Q}(A)$ for all  $c\ge 0$ and $\mu_{Q}(A+B) \le \mu_{Q}(A) +  \mu_{Q}(B)$ (positive homogeneity and subadditivity), see \cite{fiore2016contraction}. 
Based on $\mu_{Q}(A)$, the following theorem formulates a contractivity condition for PWSC systems:
 \begin{theorem}\label{theorem:PWSCcontractive}(\cite{di2014contraction})
Assume that  \eqref{eq:generalPWSnew} is PWSC and has a  unique \textit{Caratheodory} solution\footnote{See \cite{corts2008discontinuous}  for the conditions guaranteeing the existence and uniqueness of \textit{Caratheodory} and \textit{Filippov} solutions.}. Let $\mathcal{C} \subseteq  \mathbb{R}^n$ be a forward invariant set of   \eqref{eq:generalPWSnew}. If there exist $Q  \succ 0$ and $c >0$ such that:
\begin{align}\label{eq:contractivecondition}
\mu_{Q}(\frac{\partial f_i}{\partial x}(x)) \le -c
\end{align}
holds for all  $x \in \bar{\mathcal{S}}_i$, $i \in  \mathcal{N}$, then there  exist for any pair of different initial states $x_a(0), x_b(0) \in \mathcal{C}$ Caratheodory solutions $x_a(t)$ and $x_b(t)$ which satisfy:
\begin{align}\label{eq:contractivecondition2}
||x_a(t) -x_b(t)||_2 \le \alpha e^{-c t} ||x_a(0) -x_b(0)||_2
\end{align}
for an $ \alpha  >0$ and for all $t \ge 0$. The  PWSC system  \eqref{eq:generalPWSnew} is then called contractive with rate $c$ in $\mathcal{C}$.

\end{theorem}

\subsection{Filippov Convention}\label{subsec:Filipov}

For PWS systems \eqref{eq:generalPWSnew} with discontinuous vector field on the switching manifold, the \textit{Filippov convention} is often carried out to determine Filippov solutions of  \eqref{eq:generalPWSnew}, see \cite{filippov2013differential}. Let the switching manifolds $\Sigma_{i,j} \in \Omega$ separating  $\mathcal{S}_i$ and  $\mathcal{S}_j$ have co-dimension 1, and let:
\begin{align}\label{eq:Sigma}
\Sigma_{i,j}:=\{x \in  \mathbb{R}^n ~| ~H_{i,j}(x)=0 \}
\end{align}
for a smooth function $H_{i,j}(x): \mathbb{R}^n \to \mathbb{R}$, with  $\mathcal{L}_{f_i}H_{i,j}(x):= \nabla H_{i,j}(x) f_i(x)$ representing the \textit{Lie derivaive} of $H_{i,j}(x)$ with respect to $f_i(x)$.

\begin{assumption}\label{assum:lierequirement}
Given a  set  $\mathcal{C} \subseteq  \mathbb{R}^n$, the inequalities
 $\mathcal{L}_{f_i}H_{i,j}(x)\ne 0$  or $\mathcal{L}_{f_j}H_{i,j}(x) \ne 0$ hold for all $x \in \Sigma_{i,j} \cap \mathcal{C}$.
\end{assumption}
Under this assumption and with $\mathcal{S}_i \subset  \{x ~| ~H_{i,j}(x)<0 \}$ and $\mathcal{S}_j \subset  \{x ~| ~H_{i,j}(x)>0 \}$, the   switching manifold $\Sigma_{i,j}$  in $\mathcal{C}$ can be decomposed into 
 the following three regions\footnote{Any $x$ satisfying $ \mathcal{L}_{f_i}H_{i,j}(x) \ne 0$ and $  \mathcal{L}_{f_j}H_{i,j}(x)=0$, or vice versa, belongs to the \textit{Sliding region} according to \cite{kuznetsov2003one}.}:  (1) \textit{crossing region} $\Sigma_c:= \{x | \mathcal{L}_{f_i}H_{i,j}(x) \mathcal{L}_{f_j}H_{i,j}(x)  \hspace{-0.5mm} > \hspace{-0.5mm}0 \}$; (2) \textit{sliding region} $\Sigma_s := \{x  |  \mathcal{L}_{f_i}H_{i,j}(x) \hspace{-0.5mm}> \hspace{-0.5mm}0,   \mathcal{L}_{f_j}H_{i,j}(x)  \hspace{-0.5mm} < \hspace{-0.5mm}0 \}$, and  (3) \textit{escaping region}  $\Sigma_e:=\{x  |  \mathcal{L}_{f_i}H_{i,j}(x) \hspace{-0.5mm}< \hspace{-0.5mm}0,  \mathcal{L}_{f_j}H_{i,j}(x) \hspace{-0.5mm} >0 \}$.
For any $x \in \Sigma_{i,j}$, provided it is not located on any other switching manifold,  the Filippov solution of   \eqref{eq:generalPWSnew}  crosses $\Sigma_{i,j}$  when  $x$ belongs to the  \textit{Crossing region}, while sliding along $\Sigma_{i,j}$ if it belongs to the \textit{Sliding region}. The sliding vector $f^s(x)$ is determined by the  \textit{Filippov convention}:
\begin{align}\label{eq:convention}
f^s(x) \hspace{-0.3mm} = \hspace{-0.3mm}(1  \hspace{-0.3mm}-  \hspace{-0.3mm}\lambda) f_i(x)  \hspace{-0.3mm}+  \hspace{-0.3mm}\lambda f_j(x), ~\nabla H_{i,j}(x) f^s(x)  \hspace{-0.3mm}=0
\end{align}
with $\lambda  \in [0,1]$. For  $x$ in the escaping region, the Filippov solution is not unique, i.e. this case is excluded from the  upcoming analysis.

\section{PWS Systems with Non-Intersecting Switching Manifolds in  $\mathbb{R}^n$ }\label{sec:PWSPSM}

Consider  the  PWS system \eqref{eq:generalPWSnew}
with $\Omega$ containing in total $N-1$  non-intersecting switching manifolds:
\begin{align}\label{eq:SigmaPM}
\Sigma_{i,i+1}:=\{x ~ | ~ H_{i,i+1}(x)=0 \} \in \Sigma,  ~ i\in\{1,\dots, N-1\}.
\end{align}
For the special case in which $H_{i,i+1}(x)=C_{i,i+1}x - d_{i,i+1}$, the switching manifolds represent a set of parallel hyperplanes, see Fig.~\ref{fig:partition}.
 The sets $\mathcal{S}_i$ are  given by:
\begin{align*}
\mathcal{S}_i = \begin{cases}
 \{x  ~| ~H_{1,2}(x)<0 \},~i=1 \\
  \{x | H_{i-1,i}(x)>0,  H_{i,i+1}(x)<0\}, i\in\{2,\ldots,N-1\} \\
  \{x  ~| ~H_{N-1,N}(x)>0 \},~i=N.
    \end{cases}
\end{align*}
To study the contraction property of the  Filippov solution of \eqref{eq:generalPWSnew}, the method of regularization is adopted: First of all, the discontinuous vector field along each switching manifold is smoothly approximated using a transition function:
\begin{align}\label{eq:transitonfunction}
\phi(s)=  \begin{cases}
1, ~~~~\textit{if}~~s \ge 1\\
s,~~~~\textit{if }~~s\in (-1, 1) \\
-1,~~~~\textit{if}~~s \le -1
    \end{cases}
\end{align}
For this function, $\phi'(s)=\frac{\partial \phi}{\partial s} \ge 0$ holds for all three intervals $(-\infty, -1)$, $(-1, 1)$, and $(1, \infty)$, and it can be continuously extended to the closure of these intervals.

\begin{definition}
The $\phi$-regularization of the  PWS system \eqref{eq:generalPWSnew}  with the switching manifolds in \eqref{eq:SigmaPM} is defined by:
\begin{align}\label{eq:regularizeddynamics}
&\dot{x}=F_\epsilon (x)=\frac{1}{2} \bigg(\hspace{-0.7mm}  \sum\limits_{i=2}^{N-1}(\phi(\frac{H_{i-1,i}(x)}{\epsilon}) \hspace{-0.5mm} - \hspace{-0.5mm}  \phi(\frac{H_{i,i+1}(x)}{\epsilon}))f_i(x)  \notag\\
&\hspace{-0.3mm}+ \hspace{-0.8mm}( \hspace{-0.5mm}1 \hspace{-0.5mm}- \hspace{-0.5mm}\phi(\frac{H_{1,2}(\hspace{-0.3mm}x\hspace{-0.3mm})}{\epsilon}\hspace{-0.3mm})\hspace{-0.3mm})f_1 \hspace{-0.3mm}(x) \hspace{-0.7mm} + \hspace{-0.7mm} ( \hspace{-0.5mm}1 \hspace{-0.7mm}+ \hspace{-0.7mm}\phi(\frac{H_{N-1,N}(\hspace{-0.3mm}x_\epsilon\hspace{-0.3mm})}{\epsilon}\hspace{-0.3mm})\hspace{-0.3mm})f_N \hspace{-0.3mm}(x)  \hspace{-0.9mm}\bigg) 
\end{align}
for $0<\epsilon \ll 1$.
In addition, a set $\Omega_r:= \{\mathcal{S}^{1,2}_\epsilon, \ldots, \mathcal{S}^{N-1,N}_\epsilon \}$ of  \textit{regions of regularization} are  defined by:
\begin{align}\label{eq:regularizationset}
\mathcal{S}^{i,i+1}_\epsilon:= \hspace{-0.5mm}\{x | -\epsilon \hspace{-0.5mm}<\hspace{-0.5mm} H_{i,i+1}(x) \hspace{-0.5mm} < \hspace{-0.5mm}\epsilon \},~i =1,\ldots, N \hspace{-0.5mm}- \hspace{-0.5mm}1.
\end{align}
\end{definition}
It can be observed from  Fig.~\ref{fig:partition} that the state space $\mathbb{R}^n$, which was originally partitioned by $H_{i,i+1}(x)=0$ into $N$ regions in \eqref{eq:generalPWSnew}, is now partitioned by $H_{i,i+1}(x)=\pm \epsilon$ into $2N-1$ regions on which the regularized dynamics \eqref{eq:regularizeddynamics} is defined:
\begin{align*}
F_\epsilon (x) \hspace{-0.7mm}= \hspace{-0.7mm} \begin{cases}
f_1(x),~x \in \mathcal{S}_1 \setminus \mathcal{S}^{1,2}_\epsilon\\
f_i(x), x \in \mathcal{S}_i \hspace{-0.5mm} \setminus  \hspace{-0.5mm} \{\mathcal{S}^{i-1,i}_\epsilon \hspace{-0.5mm}  \cup \hspace{-0.5mm}  \mathcal{S}^{i,i+1}_\epsilon\},  i =2,\ldots, N-1 \\
f_N(x),~x \in \mathcal{S}_N \setminus \mathcal{S}^{N-1, N}_\epsilon
    \end{cases}
    \end{align*}
Since $\phi(s)$ is continuous,  the regularized dynamics \eqref{eq:regularizeddynamics} is PWSC in all $2N-1$ regions according to Def.~\ref{def:PWSC}. (The extendability of the Jacobian  $\frac{\partial f_{i}}{\partial x} (x)$ to the closure of each region is analyzed later, below \eqref{eq:Jacobiancompute}.) The dynamics has a unique Caratheodory solution  $x_{\epsilon}(t)$ according to Proposition 1 in \cite{corts2008discontinuous}.

\begin{lemma} \label{lemma:regudifference}
Let the switching manifold $\Sigma_{i,i+1}$ of the  PWS system \eqref{eq:generalPWSnew}  satisfy  Assumption~\ref{assum:lierequirement} for   a  set $\mathcal{C} \subseteq  \mathbb{R}^n$. For 
any initial state $x(0) \in \mathcal{S}^{i,i+1}_\epsilon  \subseteq \mathcal{C}$ and any time  $t_f \ge  0$,  the Filippov solution  $x(t)$ of   \eqref{eq:generalPWSnew} satisfies  $x(t) \in \mathcal{S}^{i,i+1}_\epsilon$  for all $t \in [0, t_f]$, and:
\begin{align}\label{eq:distancereduction}
\lim_{\epsilon \to 0^+}||x(t) -x_{\epsilon}(t)|| = 0
    \end{align}
 holds uniformly between  $x(t)$ and  the  Caratheodory solution $x_{\epsilon}(t)$ of  \eqref{eq:regularizeddynamics}    for all $t \in [0, t_f]$.
\end{lemma}
\begin{proof}
The proof follows the same path as in   \cite{fiore2016contraction} for Lemma 2 or  in \cite{llibre2008sliding} for Theorem 1.1. To show \eqref{eq:distancereduction}, singular perturbation analysis is carried out to cast the regularized  dynamics   \eqref{eq:regularizeddynamics}  into a slow-fast system for each switching manifold
$\Sigma_{i,i+1}$, $i  =1,\dots, N-1$. To briefly review the idea, let $x_{[l]}$, $l=1,\ldots, n$, denote the $j$-th element of $x$ and assume that:
\begin{align*}
H_{i,i+1}(x)=x_{[1]}-d_{i, i+1}
    \end{align*}
    for all $i =1,\ldots, N-1$ (see Fig.~\ref{fig:partition}, while the proof for more general $H_{i,i+1}(x)$ can be found in \cite{kaklamanos2019regularization}). By introducing the transformation:
\begin{align}\label{eq:statetransform}
\hat{x}_{[1]}:= \frac{1}{\epsilon} (x_{[1]} - d_{i, i+1}),
\end{align} 
 the state vector $x=[x_{[1]},\ldots, x_{[n]}]^T$ is cast into:
\begin{align}\label{eq:statevectortransform}
x =[\epsilon\hat{x}_{[1]} +  d_{i, i+1}, x_{[2]},\ldots, x_{[n]}]^T.
\end{align} 
Since  the regions of regularization in $\Omega_r$  do not intersect with each other, there  exists:
\begin{align}\label{eq:regularizeddynamicsequaton}
&\phi(\frac{H_{j,j+1}(x)}{\epsilon}) = \begin{cases}
1, ~\forall  j =1, \ldots, i-1 \\
-1, ~\forall  j =i+1, \ldots, N-1
    \end{cases}.
    \end{align}
for any  $x \in \mathcal{S}^{i,i+1}_\epsilon$.    
    Based on \eqref{eq:regularizeddynamicsequaton}, the regularized dynamics \eqref{eq:regularizeddynamics}  can  be cast into a  slow-fast system with:
\begin{align}\label{eq:regularizeddynamics2}
\epsilon \dot{\hat{x}}_{[1]}  \hspace{-0.8mm} = \hspace{-0.8mm} \frac{1}{2} \bigg(  \hspace{-0.9mm}(1 \hspace{-0.8mm}- \hspace{-0.8mm} \phi(\hat{x}_{[1]}))f_{i, [1]}\hspace{-0.3mm}(x)  \hspace{-0.8mm}+ \hspace{-0.8mm} (1 \hspace{-0.8mm} + \hspace{-0.8mm} \phi(\hat{x}_{[1]}))f_{i  \hspace{-0.3mm}+\hspace{-0.3mm}1, [1]}\hspace{-0.3mm}(x)  \hspace{-0.9mm} \bigg)
\end{align}
for the dynamics of the fast variables, and for the slow ones:
\begin{align}\label{eq:regularizeddynamics3}
\dot{x}_{[l]} \hspace{-0.7mm}= \hspace{-0.7mm}\frac{1}{2} \bigg(\hspace{-0.7mm}(1 \hspace{-0.7mm}- \hspace{-0.7mm}\phi(\hat{x}_{[1]}))f_{i, [l]}(x)\hspace{-0.7mm} +\hspace{-0.7mm} (1 \hspace{-0.7mm}+ \hspace{-0.7mm}\phi(\hat{x}_{[1]}))f_{i+1, [l]}(x) \hspace{-0.7mm}\bigg)
\end{align}
for all $l =2,\ldots, n$. For $\hat{x}_{[1]}$, the fixed point $\bar{\hat{x}}_{[1]}$ (i.e. the critical manifold of the boundary-layer model) of \eqref{eq:regularizeddynamics2}  at the  limit $\epsilon  \to 0^+$ must satisfy:
\begin{align}\label{eq:regularizeddynamics7}
\phi(\bar{\hat{x}}_{[1]}) = - \frac{f_{i+1, [1]}(x)  + f_{i, [1]}(x)  }{f_{i+1, [1]}(x)  - f_{i, [1]}(x)  }
\end{align}
due to Assumption~\ref{assum:lierequirement}, and thus:
\begin{align}\label{eq:fixpoints}
\bar{\hat{x}}_{[1]}=h_{i, i+1}(x_{[2]}, \ldots, x_{[n]}) 
\end{align}
must hold for some function $h_{i, i+1}: \mathbb{R}^{n-1} \to \mathbb{R}$ according to the implicit function theorem, see \cite{krantz2002implicit}. Define:
\begin{align}\label{eq:reducedvector}
x_{rd} := [h_{i, i+1}(x_{[2]}, \ldots, x_{[n]}) , x_{[2]},\ldots, x_{[n]}]^T
\end{align}
and by  inserting \eqref{eq:regularizeddynamics7} into \eqref{eq:regularizeddynamics3}, the dynamics of the \textit{reduced problem} on the  critical manifold of the boundary-layer model is given by:
\begin{align}\label{eq:regularizeddynamics8}
\dot{x}_{[l]}=&\frac{f_{i+1, [1]}(x_{rd}) }{f_{i+1, [1]}(x_{rd})  - f_{i, [1]}(x_{rd})} f_{i, [l]}(x_{rd}) \notag \\
&  + (1-\frac{f_{i+1, [1]}(x_{rd}) }{f_{i+1, [1]}(x_{rd})  - f_{i, [1]}(x_{rd})})  f_{i+1, [l]}(x_{rd})
\end{align}
for all $l =2,\ldots, n$.  Note that the dynamics  in \eqref{eq:regularizeddynamics8} coincides with the sliding dynamics  in \eqref{eq:convention} of   the Filippov solution  $x(t)$ of \eqref{eq:generalPWSnew} on the switching manifold $\Sigma_{i,i+1}$. 
For $\epsilon \to 0^+$,  since $\lim_{\epsilon \to 0^+}  \mathcal{S}^{i,i+1}_\epsilon =\Sigma_{i,i+1}$, the motion of $x_{\epsilon}(t)$  of  \eqref{eq:regularizeddynamics} on  $\Sigma_{i,i+1}$ is governed solely by the motion of the slow variables. As a result, the motions of  $x_{\epsilon}(t)$ and $x(t)$  on  $\Sigma_{i,i+1}$ are identical for  $\epsilon \to 0^+$, and therefore \eqref{eq:distancereduction}  must  hold uniformly.
\hfill$\Box$ 
\end{proof}

\begin{figure}[t!]
  \begin{center}
		\psfrag{x1}[rc][rc][1]{$x_{[1]}$}
		\psfrag{x2}[rc][rc][1]{$x_{[2]}$}
			\psfrag{e}[rc][rc][1]{$\epsilon$}
    \psfrag{h12}[rc][rc][1]{$H_{1,2}(x)=0$}			
	 \psfrag{h23}[rc][rc][1]{$H_{2,3}(x)=0$}	
	 \psfrag{h34}[rc][rc][1]{$H_{3,4}(x)=0$}							

    \psfrag{s1}[rc][rc][1]{$\mathcal{S}^{1,2}_\epsilon$}
        \psfrag{s2}[rc][rc][1]{$\mathcal{S}^{2,3}_\epsilon$}		
            \psfrag{s3}[rc][rc][1]{$\mathcal{S}^{3,4}_\epsilon$}				
					
			\psfrag{f1}[rc][rc][1]{$\dot{x}=f_1(x)$}				
			\psfrag{f2}[rc][rc][1]{$\dot{x}=f_2(x)$}		
		\psfrag{f3}[rc][rc][1]{$\dot{x}=f_3(x)$}				
		
					\psfrag{t1}[rc][rc][0.8]{$\dot{x}=f_1(x)$}				
			\psfrag{t2}[rc][rc][0.8]{$\dot{x}=f_2(x)$}		
		\psfrag{t3}[rc][rc][0.8]{$\dot{x}=f_3(x)$}					
					
 \includegraphics[width = 0.48\textwidth]{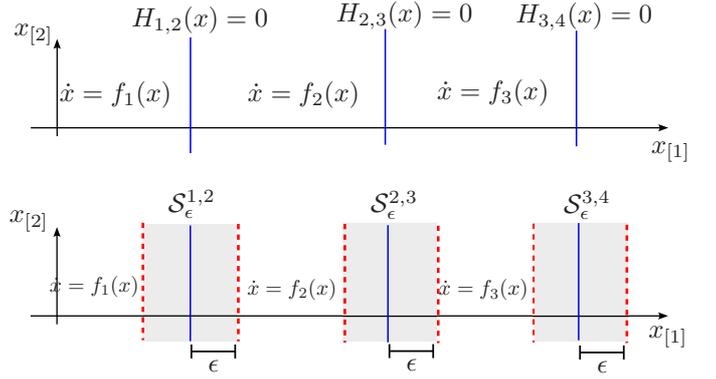}
    \caption{Each switching manifold $\Sigma_{i,i+1}$ of  \eqref{eq:generalPWSnew} is  expanded to generate the region $\mathcal{S}^{i,i+1}_\epsilon$ according to \eqref{eq:regularizationset}. }
    \label{fig:partition}
  \end{center}
\end{figure}

Based on Lemma~\ref{lemma:regudifference}, the following fact can also be established:
\begin{lemma} \label{lemma:regudifferenceglobal}
Assume that  all switching boundaries in  $\Omega$ satisfy  Assumption~\ref{assum:lierequirement}  for a common  forward invariant set  $\mathcal{C} \subseteq  \mathbb{R}^n$   of   \eqref{eq:generalPWSnew}. The Filippov solution $x(t)$ of \eqref{eq:generalPWSnew} and the Caratheodory solution $x_{\epsilon}(t)$ of   \eqref{eq:regularizeddynamics} satisfy for any initial state $x(0) \in \mathcal{C}$:
\begin{align}\label{eq:distancereductionglobal}
\lim_{\epsilon \to 0^+}||x(t) -x_{\epsilon}(t)|| = 0
    \end{align}
  uniformly for all $t \ge 0$.
\end{lemma}
\begin{proof}
According to Lemma~\ref{lemma:regudifference} the courses of   $x(t)$ and $x_{\epsilon}(t)$ coincide within each regularization region $\mathcal{S}^{i,i+1}_\epsilon \in \Omega_r$,  $i  =1,\dots, N-1$, 
as $\epsilon \to 0^+$. Based on the definition of $F_\epsilon (x)$, the trajectories are also identical  outside of the  regions of regularization, and thus they coincide everywhere in $\mathbb{R}^n$. Consequently, \eqref{eq:distancereductionglobal} holds  for all  $t \ge 0$ if the state is initialized to the forward invariant set: $x(0)\in\mathcal{C}$. \hfill$\Box$ 
\end{proof}

According to \eqref{eq:distancereductionglobal}, the contraction analysis of the  Filippov solution $x(t)$ of \eqref{eq:generalPWSnew} can be carried out by  examining the contraction properties of the regularized PWSC system  \eqref{eq:regularizeddynamics}. First,   the  Jacobian  of $F_\epsilon (x)$   is given by:
\begin{align}\label{eq:Jacobiancompute}
&\frac{\partial F_{\epsilon}}{\partial x} (x)  \hspace{-0.7mm}= \hspace{-0.7mm} \frac{1  \hspace{-0.7mm} -  \hspace{-0.7mm} \phi(\frac{H_{1,2}(x)}{\epsilon})}{2}  \frac{\partial f_{1}}{\partial x} (x)  \hspace{-0.7mm}+  \hspace{-0.7mm} \frac{1  \hspace{-0.7mm} +  \hspace{-0.7mm} \phi(\frac{H_{N  \hspace{-0.2mm}- \hspace{-0.2mm}1,N}(x)}{\epsilon})}{2} \frac{\partial f_{N}}{\partial x} (x) \notag \\
&~~~~~~~~+ \sum\limits_{i=2}^{N-1} \frac{\phi(\frac{H_{i-1,i}(x)}{\epsilon} ) -\phi(\frac{H_{i,i+1}(x)}{\epsilon} )}{2}\frac{\partial f_{i}}{\partial x} (x)  \notag \\
&~~~~~~~~+ \sum\limits_{i=1}^{N-1} \frac{\phi'(\frac{H_{i,i+1}(x)}{\epsilon} )}{2\epsilon}( f_{i+1}(x)- f_{i}(x))\nabla H(x).
\end{align}
Note that  $\frac{\partial f_{\epsilon}}{\partial x} (x)$ is not everywhere continuous  due to the term $\phi'(\frac{H_{i,i+1}(x)}{\epsilon})$. Nevertheless, it remains continuous (and non-negative) within each set  $\mathcal{S}_i  \setminus \{\bar{\mathcal{S}}^{i-1,i}_\epsilon  \cup   \bar{\mathcal{S}}^{i,i+1}_\epsilon\}$, 
as well as within each set  $\mathcal{S}^{i,i+1}_\epsilon$.
Moreover, the  Jacobian $\frac{\partial F_{\epsilon}}{\partial x} (x)$ can  be continuously extended to the closure of these sets  (see the properties of $\phi'(\cdot)$ discussed after \eqref{eq:transitonfunction}).
Thus, based on Theorem~\ref{theorem:PWSCcontractive}, the following conditions are  proposed to ensure that the solution of  \eqref{eq:regularizeddynamics} is contracting:
\begin{theorem}\label{theorem:regulcontractive}
Let $\mathcal{C} \subseteq  \mathbb{R}^n$ be a forward invariant set of  \eqref{eq:regularizeddynamics}.  If there exist  $Q \in \mathbb{R}^{n \times n} \succ 0$ and $c >0$ such that:
\begin{align}
&\mu_{Q}(  \frac{\partial f_{1}}{\partial x} (x)) \le -c,   ~~\forall x \in \mathcal{S}_1  \cup\bar{\mathcal{S}}^{1,2}_\epsilon \label{eq:contractivefinalcondition0} \\
&\mu_{Q}\hspace{-0.3mm} ( \hspace{-0.3mm} \frac{\partial f_{i}}{\partial x} (x) \hspace{-0.3mm} )  \hspace{-1.1mm}  \le  \hspace{-1mm}   -c, \hspace{-0.3mm} \forall  \hspace{-0.1mm}  x  \hspace{-0.7mm}  \in  \hspace{-0.7mm}  \bar{\mathcal{S}}^{i-1,i}_\epsilon  \hspace{-0.7mm}  \cup  \hspace{-0.7mm}  \mathcal{S}_i    \hspace{-0.7mm} \cup  \hspace{-0.7mm}  \bar{\mathcal{S}}^{i,i+1}_\epsilon\hspace{-0.3mm},  i  \hspace{-0.7mm}  =  \hspace{-0.7mm} 2  \ldots N \hspace{-1mm} -  \hspace{-0.9mm} 1 \label{eq:contractivefinalcondition1} \\
&\mu_{Q}(  \frac{\partial f_{N}}{\partial x} (x)) \le -c, \forall x \in \bar{\mathcal{S}}^{N-1,N}_\epsilon \cup \mathcal{S}_N \label{eq:contractivefinalcondition1NEW} \\
&\mu_{Q}( (f_{i+1}(x) \hspace{-0.7mm}  - \hspace{-0.7mm}  f_{i}(x))\nabla H(x) ) \hspace{-0.7mm}  \le \hspace{-0.7mm}  0,\forall x \hspace{-0.7mm}  \in \hspace{-0.7mm}  \bar{\mathcal{S}}^{i,i+1}_\epsilon,   i  \hspace{-0.7mm}  =  \hspace{-0.7mm} 1 \ldots N \hspace{-1mm} -  \hspace{-0.9mm} 1   \label{eq:contractivefinalcondition2}
\end{align}
then the regularized  PWSC dynamics \eqref{eq:regularizeddynamics} is contractive with rate $c$ in $\mathcal{C}$.
\end{theorem}
\begin{proof}
Based on the subadditivity of $\mu_{Q}$, there exists:
\begin{align}\label{eq:Jacobianmeasure}
&\mu_{Q}(\frac{\partial F_{\epsilon}}{\partial x} (x)) \le \mu_{Q}( \frac{1 -\phi(\frac{H_{1,2}(x)}{\epsilon})}{2}  \frac{\partial f_{1}}{\partial x} (x)) \notag \\
&+\mu_{Q}( \frac{1 +\phi(\frac{H_{N-1,N}(x)}{\epsilon})}{2} \frac{\partial f_{N}}{\partial x} (x)) \notag \\
&+ \sum\limits_{i=2}^{N-1} \mu_{Q}(\frac{\phi(\frac{H_{i-1,i}(x)}{\epsilon} ) -\phi(\frac{H_{i,i+1}(x)}{\epsilon} )}{2}\frac{\partial f_{i}}{\partial x} (x))  \notag \\
&+ \sum\limits_{i=1}^{N-1} \mu_{Q}(\frac{\phi'(\frac{H_{i,i+1}(x)}{\epsilon} )}{2\epsilon}( f_{i+1}(x)- f_{i}(x))\nabla H(x)).
\end{align}
As the inequalities:
\begin{align}\label{eq:Jacobianmeasure2}
&1 -\phi(\frac{H_{1,2}(x)}{\epsilon} ) \ge 0,~1 +\phi(\frac{H_{N-1,N}(x)}{\epsilon} ) \ge 0, \notag \\
&\phi(\frac{H_{i-1,i}(x)}{\epsilon} )-\phi(\frac{H_{i,i+1}(x)}{\epsilon} ) \ge 0,   i  \hspace{-0.7mm}  =  \hspace{-0.7mm} 2 \ldots N \hspace{-1mm} -  \hspace{-0.9mm} 1
\end{align}
 hold for all $x \in  \mathbb{R}^n$, and $\phi'(\cdot) \ge 0$ holds  within each set  $\mathcal{S}_i  \setminus \{\bar{\mathcal{S}}^{i-1,i}_\epsilon  \cup   \bar{\mathcal{S}}^{i,i+1}_\epsilon\}$
and  $\mathcal{S}^{i,i+1}_\epsilon$, the inequality $\mu_Q(\frac{\partial F_{\epsilon}}{\partial x} (x)) \le -c$ must apply within each of the $2N-1$ regions  of  \eqref{eq:regularizeddynamics}
and their closures. As a result and according to Theorem~\ref{theorem:PWSCcontractive}, the PWSC dynamics \eqref{eq:regularizeddynamics} must be contractive with rate $c$ in $\mathcal{C}$.  \hfill$\Box$

\end{proof}
Based on  Lemma~\ref{lemma:regudifferenceglobal} and Theorem~\ref{theorem:regulcontractive}, the following conditions are proposed to ensure that the  Filippov solutions  of   \eqref{eq:generalPWSnew} are contracting:
\begin{theorem}\label{theorem:originPWScontractive}
Assume that  all switching boundaries in \eqref{eq:SigmaPM}  satisfy  Assumption~\ref{assum:lierequirement}  for a common,  forward invariant set  $\mathcal{C} \subseteq  \mathbb{R}^n$   of  PWS system \eqref{eq:generalPWSnew}. Then, if  there exist  $Q \in \mathbb{R}^{n \times n} \succ 0$ and $c >0$ such that:
\begin{align}
& \mu_{Q}(  \frac{\partial f_{i}}{\partial x} (x)) \le -c, ~~\forall x \in \bar{\mathcal{S}}_i \cap \mathcal{C}, ~  i =1,\ldots, N \label{eq:contractivePWScondition1} \\
&  \mu_{Q}( (f_{i+1}  \hspace{-0.3mm}(x)   \hspace{-0.9mm} -   \hspace{-0.9mm} f_{i}(x))   \hspace{-0.5mm} \nabla   \hspace{-0.5mm} H(x))  \hspace{-0.9mm} \le   \hspace{-0.7mm}0,   \hspace{-0.5mm}  \forall   \hspace{-0.3mm} x   \hspace{-0.3mm} \in   \hspace{-0.7mm} \Sigma_{i,i+1}   \hspace{-0.5mm} \cap   \hspace{-0.5mm} \mathcal{C},  i  \hspace{-0.7mm}  =  \hspace{-0.7mm} 1 \ldots N \hspace{-1mm} -  \hspace{-0.9mm} 1  \label{eq:contractivePWScondition2}
\end{align}
hold,  for any two  initial states $x_a(0), x_b(0) \in \mathcal{C}$,  the respective Filippov solutions of   \eqref{eq:generalPWSnew} satisfy \eqref{eq:contractivecondition2}  for all $t \ge 0$.

\end{theorem}
\begin{proof}
Note that for  $\epsilon \to 0^+$, there exist $\lim_{\epsilon \to 0^+} \bar{\mathcal{S}}^{i-1,i}_\epsilon \cup \mathcal{S}_i   \cup\bar{\mathcal{S}}^{i,i+1}_\epsilon = \bar{\mathcal{S}}_i$ and $\lim_{\epsilon \to 0^+}\bar{\mathcal{S}}^{i,i+1}_\epsilon = \Sigma_{i,i+1}$. The Caratheodory solution $x_{\epsilon}(t)$ of \eqref{eq:regularizeddynamics}  thus must be contractive   with  rate $c$ in  $\mathcal{C}$ under  \eqref{eq:contractivePWScondition1} and  \eqref{eq:contractivePWScondition2}  according to  Theorem~\ref{theorem:regulcontractive}.
As $\lim_{\epsilon \to 0^+}||x(t) -x_{\epsilon}(t)|| = 0$ also  holds  for the   Filippov solution $x(t)$ of   \eqref{eq:generalPWSnew} according to Lemma~\ref{lemma:regudifferenceglobal}, $x(t)$ thus must also be contracting with the same rate as $x_{\epsilon}(t)$  in  $\mathcal{C}$.
 \hfill$\Box$ 
\end{proof}
\begin{remark}\label{remark:comparetoliterature}
Compared to the contraction condition for PWS systems with a single switching manifold in \cite{fiore2016contraction}, the new conditions presented here essentially replicate the former condition for each switching manifold. However, if the switching manifolds intersect, as demonstrated in the following section, the regularized dynamics and the corresponding Jacobian exhibit a different structure than before, leading to distinct contraction conditions.
\end{remark}




\section{PWS Systems with Intersecting Switching Manifolds in $\mathbb{R}^2$ }\label{sec:PWSPSM}

This section discusses how to extend the previous result to intersecting switching manifolds. For simplicity of the discussion, the PWS system \eqref{eq:generalPWSnew}  is assumed to be defined for the $ \mathbb{R}^2$
with two intersecting switching manifolds:
\begin{align}\label{eq:SigmaPMintersecting}
\Sigma_{1} \hspace{-0.5mm} := \hspace{-0.5mm} \{ \hspace{-0.5mm} x  \in \mathbb{R}^2 | H_{1}(x) \hspace{-0.5mm} = \hspace{-0.5mm} 0 \},  \Sigma_{2} \hspace{-0.5mm} := \hspace{-0.5mm} \{ \hspace{-0.5mm} x \in \mathbb{R}^2  | H_{2}(x) \hspace{-0.5mm} = \hspace{-0.5mm} 0 \} 
\end{align}
which partition  $ \mathbb{R}^2$ into four regions (see Fig.~\ref{fig:partitionintersection}):
\begin{align*}
&\dot{x}=f_1(x), x \in  \mathcal{S}_1:=\{x \in \mathbb{R}^2  ~| ~H_{1}(x)>0,  H_{2}(x)<0\}\\
&\dot{x}=f_2(x), x \in \mathcal{S}_2:=\{x \in \mathbb{R}^2  ~| ~H_{1}(x)>0,  H_{2}(x)>0\}\\
&\dot{x}=f_3(x), x \in \mathcal{S}_3:=\{x \in \mathbb{R}^2 ~| ~H_{1}(x)<0,  H_{2}(x)>0\}\\
&\dot{x}=f_4(x), x \in \mathcal{S}_4:=\{x \in \mathbb{R}^2  ~| ~H_{1}(x)<0,  H_{2}(x)<0\}.
\end{align*}
It is assumed that $f_1(x),f_2(x),f_3(x), f_4(x)$ are also defined on $\Sigma_{1}$ and $\Sigma_{2}$. (Note that extensions of the considered settings to 
higher state dimensions and more  switching manifolds are provided at the end of this section.)

\begin{assumption}\label{assum:interior}
The two  switching manifolds $\Sigma_{1}$ and $\Sigma_{2}$ intersect at the unique point $\tilde{x} \in \mathbb{R}^2$, while any vector within the convex hull formed by  $f_1(\tilde{x})$, $f_2(\tilde{x})$, $f_3(\tilde{x})$, and $f_4(\tilde{x})$ point commonly into one of the sets  $\mathcal{S}_i$.
\end{assumption}
This assumption ensures that $\tilde{x}$ belongs uniquely to the crossing region through which $\mathcal{S}_i$ is reached. Otherwise, as shown in \cite{kaklamanos2019regularization, jeffrey2014dynamics}, the Filippov convention at the intersection of switching manifolds does not yield a unique sliding vector.
\begin{definition}
The $\phi$-regularization of \eqref{eq:generalPWSnew} with the switching manifolds in \eqref{eq:SigmaPMintersecting} by using the transition function $\phi$
in \eqref{eq:transitonfunction} is given by:
\begin{align}\label{eq:regularizeddynamicsintersection}
\dot{x}=F_\epsilon (x)&=\frac{1}{4}\bigg((1 + \phi(\frac{H_{1}(x)}{\epsilon}))(1 - \phi(\frac{H_{2}(x)}{\epsilon}))f_1(x) \notag\\
&+(1 + \phi(\frac{H_{1}(x)}{\epsilon}))(1 + \phi(\frac{H_{2}(x)}{\epsilon}))f_2(x) \notag\\
&+(1 - \phi(\frac{H_{1}(x)}{\epsilon}))(1 + \phi(\frac{H_{2}(x)}{\epsilon}))f_3(x) \notag\\
&+(1 - \phi(\frac{H_{1}(x)}{\epsilon}))(1 - \phi(\frac{H_{2}(x)}{\epsilon}))f_4(x) \bigg)
\end{align}
for $0<\epsilon \ll 1$.
In addition, a set $\Omega_r:= \{\mathcal{S}^{1}_\epsilon, \mathcal{S}^{2}_\epsilon\}$ of the \textit{regions of regularization} is defined by:
\begin{align}\label{eq:regularizationsetintersect}
\mathcal{S}^{i}_\epsilon:= \{x \in \mathbb{R}^2   ~| ~-\epsilon < H_{i}(x)<\epsilon \},~~i \in \{1, 2\}.
\end{align}
\end{definition}
One can notice from  Fig.~\ref{fig:partitionintersection} that the $\mathbb{R}^2 $  is partitioned into 9 regions by $ H_{i}(x)= \pm \epsilon$.
Similar to \eqref{eq:regularizeddynamics},  the relation $F_\epsilon (x) = f_i(x)$  holds for all $x \in \mathcal{S}_i \setminus \{\mathcal{S}^{1}_\epsilon \cup \mathcal{S}^{2}_\epsilon\}$, $i \in \{1,2,3,4\}$, and
the regularized dynamics  \eqref{eq:regularizeddynamicsintersection} is again PWSC with a unique Caratheodory solution  $x_{\epsilon}(t)$.

\begin{figure}[t!]
  \begin{center}
		\psfrag{x1}[rc][rc][1]{$x_{[1]}$}
		\psfrag{x2}[rc][rc][1]{$x_{[2]}$}
		\psfrag{h1}[rc][rc][1]{$H_{1}(x)=0$}
		\psfrag{h2}[rc][rc][1]{$H_{2}(x)=0$}		
		
			\psfrag{e}[rc][rc][1]{$\epsilon$}
					\psfrag{x}[rc][rc][1]{ $\tilde{x}$}
				\psfrag{e1}[rc][rc][1]{$\mathcal{S}^{1}_\epsilon$}		
						\psfrag{e2}[rc][rc][1]{$\mathcal{S}^{2}_\epsilon$}			

    \psfrag{s1}[rc][rc][1]{$\dot{x}=f_1(x)$}
        \psfrag{s2}[rc][rc][1]{$\dot{x}=f_2(x)$}		
            \psfrag{s3}[rc][rc][1]{$\dot{x}=f_3(x)$}				
	            \psfrag{s4}[rc][rc][1]{$\dot{x}=f_4(x)$}

			\psfrag{f2}[rc][rc][1]{$\dot{x}=f_2(x)$}		
		\psfrag{f3}[rc][rc][1]{$\dot{x}=f_3(x)$}				
			\psfrag{f4}[rc][rc][1]{$\dot{x}=f_4(x)$}		
				
					\psfrag{t1}[rc][rc][0.8]{$\dot{x}=f_1(x)$}				
			\psfrag{t2}[rc][rc][0.8]{$\dot{x}=f_2(x)$}		
		\psfrag{t3}[rc][rc][0.8]{$\dot{x}=f_3(x)$}					
			\psfrag{t4}[rc][rc][0.8]{$\dot{x}=f_4(x)$}					
 \includegraphics[width = 0.45\textwidth]{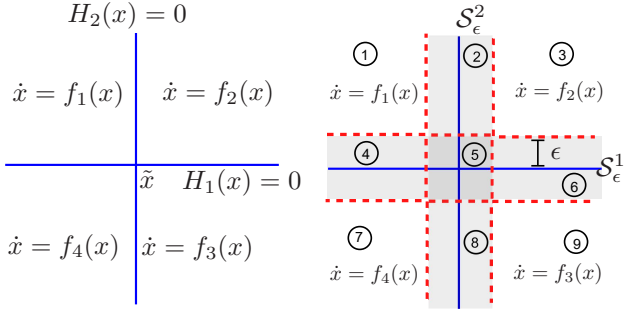}
    \caption{Both switching manifolds $\Sigma_{1}$ and $\Sigma_{2}$ are  expanded to the regions $\mathcal{S}^{1}_\epsilon$ and $\mathcal{S}^{2}_\epsilon$ according to \eqref{eq:regularizationsetintersect}, resulting in a partition of the space into 9 regions.}
    \label{fig:partitionintersection}
  \end{center}
\end{figure}

\begin{lemma} \label{lemma:regudifferenceint}
Let Assumption~\ref{assum:interior} hold and the two switching manifolds in   \eqref{eq:SigmaPMintersecting}   satisfy  Assumption~\ref{assum:lierequirement} for   a set $\mathcal{C} \subseteq  \mathbb{R}^2$. For 
any initial state $x(0) \in  \mathcal{S}^{1}_\epsilon \cup \mathcal{S}^{2}_\epsilon$ and any time  $t_f \ge  0$,  the Filippov solution  of  \eqref{eq:generalPWSnew} satisfies  $x(t) \in  \mathcal{S}^{1}_\epsilon \cup \mathcal{S}^{2}_\epsilon$ for all $t \in [0, t_f]$, and the difference of $x(t)$ and  the  Caratheodory solution $x_{\epsilon}(t)$ of   \eqref{eq:regularizeddynamicsintersection} satisfies: 
\begin{align}\label{eq:distancereductionint}
\lim_{\epsilon \to 0^+}||x(t) -x_{\epsilon}(t)|| = 0
\end{align}
for all $t \in [0, t_f]$.
\end{lemma}
\begin{proof}
For any $x \in \mathcal{S}^{i}_\epsilon \setminus  \{ \bar{\mathcal{S}}^{1}_\epsilon \cap \bar{\mathcal{S}}^{2}_\epsilon \}$, $i \in \{1,2,3,4\}$,  assume, for instance,  that $x$ is located in the region $\textcircled{4}$ in Fig.~\ref{fig:partitionintersection}. It is known that $\phi(\frac{H_{2}(x)}{\epsilon}) =-1$ holds according to \eqref{eq:transitonfunction}, and thus casts  \eqref{eq:regularizeddynamicsintersection} into:
\begin{align}\label{eq:intersectionR4}
\dot{x} \hspace{-0.5mm} = \hspace{-0.5mm}\frac{1}{2}(1\hspace{-0.5mm} + \hspace{-0.5mm}\phi(\frac{H_{1}(x)}{\epsilon}))f_1(x)\hspace{-0.5mm} +\hspace{-0.5mm} \frac{1}{2}(1 \hspace{-0.5mm} - \hspace{-0.5mm} \phi(\frac{H_{1}(x)}{\epsilon}))f_4(x).
\end{align}
Then, one can use a similar approach as in the proof of  Lemma~\ref{lemma:regudifference}, to reformulate \eqref{eq:intersectionR4} into the corresponding   slow-fast system as in \eqref{eq:regularizeddynamics2} and \eqref{eq:regularizeddynamics3}: The change of the slow variables of   \eqref{eq:intersectionR4} on the critical manifold  of the respective boundary-layer model
coincides  with the sliding dynamics  \eqref{eq:convention} of   the Filippov solution  $x(t)$ of \eqref{eq:generalPWSnew} on the switching manifold $\Sigma_{1}$. This procedure is repeatedly applied to all regions  $\textcircled{2}$,  $\textcircled{4}$,  $\textcircled{6}$, and  $\textcircled{8}$ in 
Fig.~\ref{fig:partitionintersection}. Since
$\lim_{\epsilon \to 0^+}  \mathcal{S}^{i}_\epsilon \setminus  \{ \bar{\mathcal{S}}^{1}_\epsilon \cap \bar{\mathcal{S}}^{2}_\epsilon \} =  \Sigma_{i} \setminus \{ \tilde{x}\}$ holds,  the sliding dynamics  \eqref{eq:convention}  coincides with the motion of $x_{\epsilon}(t)$ of   \eqref{eq:regularizeddynamicsintersection} everywhere on 
$ \{ \Sigma_{1} \cup \Sigma_{2} \}\setminus \{ \tilde{x}\}$  at the limit $\epsilon \to 0^+$. Furthermore, since $\tilde{x}$  is excluded from the sliding region by Assumption~\ref{assum:interior}, the relation \eqref{eq:distancereductionint} between    $x(t)$  and $x_{\epsilon}(t)$  must hold uniformly on $\Sigma_{1} \cup \Sigma_{2}$.
 \hfill$\Box$ 
\end{proof}

\begin{lemma} \label{lemma:regudifferenceglobal2}
Let Assumption~\ref{assum:interior} hold and assume that the switching manifolds  in \eqref{eq:SigmaPMintersecting}    satisfy  Assumption~\ref{assum:lierequirement} for   a forward invariant  set $\mathcal{C} \subseteq  \mathbb{R}^2$ of   \eqref{eq:generalPWSnew}.  For any initial state $x(0) \in \mathcal{C}$, the Filippov solution $x(t)$ of \eqref{eq:generalPWSnew} and the Caratheodory solution $x_{\epsilon}(t)$ of   \eqref{eq:regularizeddynamicsintersection} then satisfy:
\begin{align}\label{eq:distancereductionglobalinter}
\lim_{\epsilon \to 0^+}||x(t) -x_{\epsilon}(t)|| = 0
    \end{align}
  uniformly for all $t \ge 0$.
\end{lemma}
The proof of this result follows a similar path as that of Lemma~\ref{lemma:regudifferenceglobal}.

Based on \eqref{eq:distancereductionglobalinter},  one can again investigate the contraction property of the Filippov solution  of \eqref{eq:generalPWSnew} by analyzing the solution of  \eqref{eq:regularizeddynamicsintersection}. The 
Jacobian  of  \eqref{eq:regularizeddynamicsintersection}: 
\begin{align}\label{eq:Jacobiancompute4region}
&\frac{\partial F_{\epsilon}}{\partial x} \hspace{-0.5mm}(x) \hspace{-0.8mm} =\hspace{-0.9mm} \frac{1}{4\epsilon} \hspace{-0.5mm} \bigg(\hspace{-0.7mm}\phi'(\hspace{-0.7mm}\frac{H_{1}\hspace{-0.5mm}(x)}{\epsilon}\hspace{-0.5mm}) (f_1\hspace{-0.3mm}(x)\hspace{-0.7mm}+\hspace{-0.9mm}f_2\hspace{-0.3mm}(x) \hspace{-0.7mm} - \hspace{-0.8mm}f_3\hspace{-0.3mm}(x) \hspace{-0.8mm}- \hspace{-0.8mm}f_4\hspace{-0.3mm}(x)\hspace{-0.3mm})\nabla H_1\hspace{-0.5mm}(x) \notag \\
& +\phi'(\frac{H_{2}(x)}{\epsilon}) (f_2(x)+f_3(x)-f_1(x)-f_4(x))\nabla H_2(x)\notag  \\
&+ \hspace{-0.7mm} \phi' \hspace{-0.3mm} (\frac{H_{1}\hspace{-0.3mm}(x)}{\epsilon}) \phi(\frac{H_{2}\hspace{-0.3mm}(x)}{\epsilon})(f_2\hspace{-0.3mm}(x) \hspace{-0.7mm}+ \hspace{-0.7mm}f_4\hspace{-0.3mm}(x) \hspace{-0.7mm}- \hspace{-0.7mm}f_3\hspace{-0.3mm}(x) \hspace{-0.7mm}- \hspace{-0.7mm}f_1\hspace{-0.3mm}(x))\nabla \hspace{-0.3mm} H_1\hspace{-0.3mm}(x) \notag \\
&+ \hspace{-0.7mm} \phi' \hspace{-0.3mm} (\frac{H_{2}\hspace{-0.3mm}(x)}{\epsilon} \hspace{-0.3mm}) \phi(\frac{H_{1}\hspace{-0.3mm}(x)}{\epsilon}\hspace{-0.3mm})(f_2\hspace{-0.3mm}(x) \hspace{-0.7mm}+ \hspace{-0.7mm}f_4\hspace{-0.3mm}(x) \hspace{-0.7mm}- \hspace{-0.7mm}f_3\hspace{-0.3mm}(x) \hspace{-0.7mm}- \hspace{-0.7mm}f_1\hspace{-0.3mm}(x)\hspace{-0.3mm})\nabla \hspace{-0.3mm} H_2\hspace{-0.3mm}(x)  \hspace{-0.9mm}\bigg) \notag \\
&+\frac{1}{4}\bigg((1 + \phi(\frac{H_{1}(x)}{\epsilon}))(1 - \phi(\frac{H_{2}(x)}{\epsilon}))\frac{\partial f_{1}}{\partial x} (x) \notag\\
&+(1 + \phi(\frac{H_{1}(x)}{\epsilon}))(1 + \phi(\frac{H_{2}(x)}{\epsilon}))\frac{\partial f_{2}}{\partial x} (x) \notag\\
&+(1 - \phi(\frac{H_{1}(x)}{\epsilon}))(1 + \phi(\frac{H_{2}(x)}{\epsilon}))\frac{\partial f_{3}}{\partial x} (x) \notag\\
&+(1 - \phi(\frac{H_{1}(x)}{\epsilon}))(1 - \phi(\frac{H_{2}(x)}{\epsilon}))\frac{\partial f_{4}}{\partial x} (x)\bigg),
\end{align}
can be  continuously extended to the closure of each of the 9 regions in Fig.~\ref{fig:partitionintersection}. In addition, the inequalities:
\begin{align}\label{eq:Jacobianmeasure4region}
&(1 + \phi(\frac{H_{1}(x)}{\epsilon}))(1 - \phi(\frac{H_{2}(x)}{\epsilon})) \ge 0,  \notag \\
&(1 + \phi(\frac{H_{1}(x)}{\epsilon}))(1 +\phi(\frac{H_{2}(x)}{\epsilon})) \ge 0, \notag \\
&(1 - \phi(\frac{H_{1}(x)}{\epsilon}))(1 + \phi(\frac{H_{2}(x)}{\epsilon})) \ge 0, \notag \\
&(1 - \phi(\frac{H_{1}(x)}{\epsilon}))(1 - \phi(\frac{H_{2}(x)}{\epsilon})) \ge 0
\end{align}
  hold for all $x \in  \mathbb{R}^2$ according to \eqref{eq:transitonfunction}. There also exist:
  \begin{align}\label{eq:regionanalyis}
 &\phi'(\frac{H_{1}(x)}{\epsilon}) =0,\forall x \notin \bar{\mathcal{S}}^{1}_\epsilon,~~  \phi'(\frac{H_{2}(x)}{\epsilon}) =0, \forall x \notin \bar{\mathcal{S}}^{2}_\epsilon .
  \end{align}

\begin{theorem}\label{theorem:regulcontractive4region}
Let $\mathcal{C} \subseteq  \mathbb{R}^2$ be a forward invariant set of  \eqref{eq:regularizeddynamicsintersection}.  If there exist  $Q \in \mathbb{R}^{2 \times 2} \succ 0$ and $c >0$ such that:
\begin{align}
& \mu_Q(  \frac{\partial f_{1}}{\partial x} (x))\hspace{-0.5mm} \le \hspace{-0.5mm}-c, \forall  x \in  \{H_{1}(x) \ge -\epsilon, H_{2}(x) \le \epsilon\}\label{eq:contractive4regioncondition1} \\
& \mu_Q(  \frac{\partial f_{2}}{\partial x} (x)) \hspace{-0.5mm}\le \hspace{-0.5mm}-c, \forall   x \in  \{H_{1}(x) \ge -\epsilon,H_{2}(x) \ge -\epsilon\}\label{eq:contractive4regioncondition2} \\
& \mu_Q(  \frac{\partial f_{3}}{\partial x} (x))\hspace{-0.5mm} \le \hspace{-0.5mm}-c, \forall   x \in  \{H_{1}(x) \le \epsilon,H_{2}(x) \ge -\epsilon\}\label{eq:contractive4regioncondition3} \\
& \mu_Q(  \frac{\partial f_{4}}{\partial x} (x)) \le -c, \forall   x \in  \{H_{1}(x) \le \epsilon, H_{2}(x) \le \epsilon\}\label{eq:contractive4regioncondition4} \\
&  \mu_Q( \hspace{-0.3mm}(f_{1}(x)\hspace{-0.8mm} +\hspace{-0.9mm}f_{2}(x)\hspace{-0.8mm}-\hspace{-0.9mm}f_{3}(x)\hspace{-0.8mm}-\hspace{-0.9mm}f_{4}(x) \hspace{-0.3mm})\nabla  \hspace{-0.3mm}H_1(x) \hspace{-0.3mm})\hspace{-0.8mm} \le\hspace{-0.8mm}0,  \hspace{-0.3mm}\forall   x\hspace{-0.5mm} \in \hspace{-0.5mm}\bar{\mathcal{S}}^{1}_\epsilon \label{eq:contractive4regioncondition5}\\
&  \mu_Q( \hspace{-0.3mm}(f_{2}(x)\hspace{-0.8mm} +\hspace{-0.9mm}f_{3}(x)\hspace{-0.8mm}-\hspace{-0.9mm}f_{1}(x)\hspace{-0.8mm}-\hspace{-0.9mm}f_{4}(x) \hspace{-0.3mm})\nabla  \hspace{-0.3mm}H_2(x) \hspace{-0.3mm})\hspace{-0.8mm} \le \hspace{-0.8mm}0,  \hspace{-0.3mm}\forall   x\hspace{-0.5mm} \in \hspace{-0.5mm}\bar{\mathcal{S}}^{2}_\epsilon \label{eq:contractive4regioncondition6}\\
&  \mu_Q( \hspace{-0.3mm}(f_{2}(x)\hspace{-0.8mm} +\hspace{-0.9mm}f_{4}(x)\hspace{-0.8mm}-\hspace{-0.9mm}f_{1}(x)\hspace{-0.8mm}-\hspace{-0.9mm}f_{3}(x) \hspace{-0.3mm})\nabla  \hspace{-0.3mm}H_1(x) \hspace{-0.3mm})\hspace{-0.8mm}  \le \hspace{-0.8mm}0,  \hspace{-0.3mm}\forall   x\hspace{-0.5mm} \in \bar{\textcircled{6}}  \label{eq:contractive4regioncondition7}\\
&  \mu_Q( \hspace{-0.3mm}(f_{1}(x)\hspace{-0.8mm} +\hspace{-0.9mm}f_{3}(x)\hspace{-0.8mm}-\hspace{-0.9mm}f_{2}(x)\hspace{-0.8mm}-\hspace{-0.9mm}f_{4}(x) \hspace{-0.3mm})\nabla  \hspace{-0.3mm}H_1(x) \hspace{-0.3mm})\hspace{-0.8mm}  \le \hspace{-0.8mm}0,  \hspace{-0.3mm}\forall   x\hspace{-0.5mm} \in \bar{\textcircled{4}}  \label{eq:contractive4regioncondition72}\\
&  \mu_Q( \hspace{-0.3mm}(f_{2}(x)\hspace{-0.8mm} +\hspace{-0.9mm}f_{4}(x)\hspace{-0.8mm}-\hspace{-0.9mm}f_{1}(x)\hspace{-0.8mm}-\hspace{-0.9mm}f_{3}(x) \hspace{-0.3mm})\nabla  \hspace{-0.3mm}H_2(x) \hspace{-0.3mm})\hspace{-0.8mm} \le \hspace{-0.8mm}0,  \hspace{-0.3mm}\forall   x\hspace{-0.5mm} \in \bar{\textcircled{2}}  \label{eq:contractive4regioncondition8}\\
&  \mu_Q( \hspace{-0.3mm}(f_{1}(x)\hspace{-0.8mm} +\hspace{-0.9mm}f_{3}(x)\hspace{-0.8mm}-\hspace{-0.9mm}f_{2}(x)\hspace{-0.8mm}-\hspace{-0.9mm}f_{4}(x) \hspace{-0.3mm})\nabla  \hspace{-0.3mm}H_2(x) \hspace{-0.3mm})\hspace{-0.8mm} \le \hspace{-0.8mm}0,  \hspace{-0.3mm}\forall   x\hspace{-0.5mm} \in \bar{\textcircled{8}}  \label{eq:contractive4regioncondition82}\\
&f_{2}(x)+f_{4}(x)=f_{1}(x)+f_{3}(x), \forall  x \in  \bar{\textcircled{5}} \label{eq:contractive4regioncondition9}
\end{align}
hold, 
then the regularized  PWSC dynamics \eqref{eq:regularizeddynamicsintersection} is contractive with rate $c$ in $\mathcal{C}$.
\end{theorem}

\begin{proof}
Similar to the proof of  Theorem.~\ref{theorem:regulcontractive}, the goal here is again to show that $\mu_Q( \frac{\partial F_{\epsilon}}{\partial x} (x)) \le -c$ holds for the PWSC system \eqref{eq:regularizeddynamicsintersection}. Note that  $\phi'(\frac{H_{1}(x)}{\epsilon}) \ge 0$ and $\phi'(\frac{H_{2}(x)}{\epsilon}) \ge 0$ hold  within each of the 
 9 regions in Fig.~\ref{fig:partitionintersection}. According to the inequalities \eqref{eq:Jacobianmeasure4region} and the
  positive homogeneity and subadditivity of $\mu_{Q}$, the conditions \eqref{eq:contractive4regioncondition1} - \eqref{eq:contractive4regioncondition6} ensure that:
\begin{align*}
&\mu_Q( \frac{\partial F_{\epsilon}}{\partial x} (x)) \le -c +c_1 +c_2,~\forall x \in \mathbb{R}^{2}
\end{align*}
 where $c_1:=\mu_Q(\phi' \hspace{-0.3mm}(\frac{H_{1}\hspace{-0.3mm}(x)}{\epsilon}) \phi(\frac{H_{2}\hspace{-0.3mm}(x)}{\epsilon})(f_2\hspace{-0.3mm}(x) \hspace{-0.7mm}+ \hspace{-0.7mm}f_4\hspace{-0.3mm}(x) \hspace{-0.7mm}- \hspace{-0.7mm}f_3\hspace{-0.3mm}(x) \hspace{-0.7mm}- \hspace{-0.7mm}f_1\hspace{-0.3mm}(x))\nabla \hspace{-0.3mm} H_1\hspace{-0.3mm}(x))$ and
$c_2:=\mu_Q(\phi' \hspace{-0.3mm} (\frac{H_{2}\hspace{-0.3mm}(x)}{\epsilon}) \phi(\frac{H_{1}\hspace{-0.3mm}(x)}{\epsilon})(f_2\hspace{-0.3mm}(x) \hspace{-0.7mm}+ \hspace{-0.7mm}f_4\hspace{-0.3mm}(x) \hspace{-0.7mm}- \hspace{-0.7mm}f_3\hspace{-0.3mm}(x) \hspace{-0.7mm}- \hspace{-0.7mm}f_1\hspace{-0.3mm}(x))\nabla \hspace{-0.3mm} H_2\hspace{-0.3mm}(x))$.  As the signs of $\phi'(\frac{H_{1}\hspace{-0.3mm}(x)}{\epsilon}) \phi(\frac{H_{2}\hspace{-0.3mm}(x)}{\epsilon})$ and $\phi'(\frac{H_{2}\hspace{-0.3mm}(x)}{\epsilon}) \phi(\frac{H_{1}\hspace{-0.3mm}(x)}{\epsilon})$ are not always definite, the values of $c_1$ and $c_2$ are separately considered for each region:
\begin{itemize}
\item In the regions  $\textcircled{1}$,  $\textcircled{3}$,  $\textcircled{7}$,  $\textcircled{8}$, i.e., for  $x \notin \bar{\mathcal{S}}^{1}_\epsilon \cup \bar{\mathcal{S}}^{2}_\epsilon$,   $c_1=c_2=0$ hold due to  \eqref{eq:regionanalyis}.
\item In the region  $\textcircled{2}$,  $\phi'(\frac{H_{1}(x)}{\epsilon}) = 0$, $\phi(\frac{H_{1}(x)}{\epsilon}) = 1$, $\phi'(\frac{H_{2}(x)}{\epsilon}) = 1$, and $\phi(\frac{H_{2}(x)}{\epsilon}) \in (-1,1)$ are satisfied. Thus $c_1=0$ applies, while  $c_2\le 0$ holds under the condition \eqref{eq:contractive4regioncondition8}.
\item In the region  $\textcircled{8}$, $\phi'(\frac{H_{1}(x)}{\epsilon}) = 0$, $\phi(\frac{H_{1}(x)}{\epsilon}) = -1$, $\phi'(\frac{H_{2}(x)}{\epsilon}) = 1$, and $\phi(\frac{H_{2}(x)}{\epsilon}) \in (-1,1)$ apply. Thus, $c_1=0$ and $c_2\le 0$ hold under the condition \eqref{eq:contractive4regioncondition82}.
\item Similar to the two previous cases, in the region $\textcircled{6}$, $c_2=0$ applies, while  $c_1\le 0$ holds under the condition \eqref{eq:contractive4regioncondition7}. In the region  $\textcircled{4}$, $c_2=0$ applies, while  $c_1\le 0$ holds under the condition \eqref{eq:contractive4regioncondition72}.
\item In the region  $\textcircled{5}$,    $\phi'(\frac{H_{1}(x)}{\epsilon}) = 1$, $\phi(\frac{H_{1}(x)}{\epsilon}) \in (-1,1)$, $\phi'(\frac{H_{2}(x)}{\epsilon}) = 1$, and $\phi(\frac{H_{2}(x)}{\epsilon}) \in (-1,1)$, 
$c_1\le 0$ and $c_2\le 0$ hold under the condition \eqref{eq:contractive4regioncondition9}.

\end{itemize}
 As a result,  $\mu_Q(\frac{\partial F_{\epsilon}}{\partial x} (x)) \le -c$  are satisfied  within each of the 9 regions
and their closures, implying that the  PWSC dynamics  \eqref{eq:regularizeddynamicsintersection}  is  contractive with rate $c$ in $\mathcal{C}$ according to Theorem~\ref{theorem:PWSCcontractive}.  \hfill$\Box$ 
\end{proof}

Finally, based on the fact that $\lim_{\epsilon \to 0^+}\bar{\mathcal{S}}^{i}_\epsilon =\Sigma_i$ and $\lim_{\epsilon \to 0^+}||x(t) -x_{\epsilon}(t)|| = 0$, the  Filippov solution  of   \eqref{eq:generalPWSnew} with switching manifolds in   \eqref{eq:SigmaPMintersecting} is contracting according to the following result:
\begin{theorem}\label{theorem:originPWScontractiveint}

Let Assumption~\ref{assum:interior} hold and  let both switching manifolds in   \eqref{eq:SigmaPMintersecting} satisfy  Assumption~\ref{assum:lierequirement}  for a common,  forward invariant set  $\mathcal{C} \subseteq  \mathbb{R}^2$   of   \eqref{eq:generalPWSnew}. If  there exist  $Q \in \mathbb{R}^{2 \times 2} \succ 0$ and $c >0$ such that:
\begin{align}
& \mu_{Q}(  \frac{\partial f_{i}}{\partial x} (x)) \le -c, ~\forall x \in \bar{\mathcal{S}}_i  \cap \mathcal{C},~\forall i = 1, 2, 3, 4, \label{eq:contractivePWScondition14region1} \\
&   \mu_Q( \hspace{-0.3mm}(f_{1}(x)\hspace{-0.8mm} +\hspace{-0.9mm}f_{2}(x)\hspace{-0.8mm}-\hspace{-0.9mm}f_{3}(x)\hspace{-0.8mm}-\hspace{-0.9mm}f_{4}(x) \hspace{-0.3mm})\nabla  \hspace{-0.3mm}H_1(x) \hspace{-0.3mm})\hspace{-0.8mm} \le\hspace{-0.8mm}0, \forall x \hspace{-0.6mm} \in\hspace{-0.6mm}\Sigma_1 \hspace{-0.6mm}\cap \hspace{-0.6mm}\mathcal{C},\label{eq:contractivePWScondition14region2} \\
& \mu_Q( \hspace{-0.3mm}(f_{2}(x)\hspace{-0.8mm} +\hspace{-0.9mm}f_{3}(x)\hspace{-0.8mm}-\hspace{-0.9mm}f_{1}(x)\hspace{-0.8mm}-\hspace{-0.9mm}f_{4}(x) \hspace{-0.3mm})\nabla  \hspace{-0.3mm}H_2(x) \hspace{-0.3mm})\hspace{-0.8mm} \le \hspace{-0.8mm}0, \forall x \hspace{-0.6mm}\in \hspace{-0.6mm}\Sigma_2 \hspace{-0.6mm}\cap \hspace{-0.6mm}\mathcal{C}, \label{eq:contractivePWScondition14region3} \\
&  \mu_Q( (f_{2}(x)+f_{4}(x)-f_{1}(x)-f_{3}(x) )\nabla H_1(x) ) \le 0, \notag\\
& \qquad \qquad \qquad \qquad \qquad \forall   x\in  \Sigma_1 \cap \{H_2(x) >0\} \cap\mathcal{C}, \label{eq:contractivePWScondition14region31}\\
&  \mu_Q( (f_{1}(x)+f_{3}(x)-f_{2}(x)-f_{4}(x) )\nabla  H_1(x) ) \le 0,  \notag\\
& \qquad \qquad \qquad \qquad \qquad \forall   x\in  \Sigma_1 \cap \{H_2(x) <0\} \cap\mathcal{C}, \label{eq:contractivePWScondition14region32}\\
&    \mu_Q( (f_{2}(x)+f_{4}(x)-f_{1}(x)-f_{3}(x) )\nabla H_2(x) ) \le 0,  \notag\\
& \qquad \qquad \qquad \qquad \qquad \forall   x\in  \Sigma_2 \cap \{H_1(x) >0\} \cap\mathcal{C}, \label{eq:contractivePWScondition14region33}\\
&   \mu_Q( (f_{1}(x)+f_{3}(x)-f_{2}(x)-f_{4}(x) )\nabla  H_2(x) ) \le 0,  \notag\\
& \qquad \qquad \qquad \qquad \qquad \forall   x\in  \Sigma_2 \cap \{H_1(x) <0\} \cap\mathcal{C}, \label{eq:contractivePWScondition14region34}\\
& f_{2}(x)+f_{4}(x)=f_{1}(x)+f_{3}(x), ~ \textit{for} ~  x =\tilde{x} \label{eq:contractivePWScondition14region4}
\end{align}
hold, then the respective Filippov solutions of   \eqref{eq:generalPWSnew} satisfy \eqref{eq:contractivecondition2} for any pair of  initial states $x_a(0), x_b(0) \in \mathcal{C}$ and for all $t \ge 0$.
\end{theorem}

Note that both Theorems \ref{theorem:originPWScontractive}  and \ref{theorem:originPWScontractiveint}  stem from the observation that the sliding dynamics (as determined by the Filippov convention  in \eqref{eq:convention} on each  switching manifold)  coincides with change of the slow variables in the  corresponding singular perturbed problem. However, if two or more switching manifolds intersect, the  sliding dynamics  becomes non-unique (see \cite{jeffrey2014dynamics}) and thus fails to follow the same change as the slow variables\footnote{For PWS systems defined in $\mathbb{R}^3$, the work in \cite{kaklamanos2019regularization} introduced a type of sliding dynamics on the  intersection of two switching manifolds with co-dimension 1, which   coincides with the motion of the slow variables  in the   singular perturbed problem.}. Theorem \ref{theorem:originPWScontractiveint}  addresses this problem by excluding  the unique intersection point $\tilde{x}$   from  the sliding region.  Following this approach, consider the more general case of  \eqref{eq:generalPWSnew} defined in 
$\mathbb{R}^n$ involving   $N$ disjoint open sets  $\mathcal{S}_i$, $i =1, \ldots, N$, and a set  $\Omega$ of possibly intersecting codimension-one switching manifolds. When a control input  $u \in  \mathbb{R}^m$ is present, i.e., $\dot{x}=f_i(x, u)$,  $x \in \mathcal{S}_i$, a contracting  state feedback control law $u=g_i(x)$ can be synthesized by enforcing: (1)  the corresponding conditions   in \eqref{eq:contractivePWScondition14region1} to \eqref{eq:contractivePWScondition14region4}, and (2) that the intersection  of any two or more  switching manifolds in    $\mathcal{C}$ of the controlled system belongs exclusively to the crossing  region.

\section{Numerical Example}
\subsection{Example 1}
Consider a PWS system of type \eqref{eq:generalPWSnew} with $N=3$ given by:
\begin{align*}
    f_1(x)&=\begin{bmatrix}-1&1\\0&-1\end{bmatrix}x+\begin{bmatrix}3\\0\end{bmatrix}, ~~ f_2(x)=\begin{bmatrix}-2&1\\0&-1\end{bmatrix}x+\begin{bmatrix}1\\0\end{bmatrix}\\
    &\qquad \qquad f_3(x) =\begin{bmatrix}-3&1\\0&-1\end{bmatrix}x +\begin{bmatrix}-2\\0\end{bmatrix}
\end{align*}
for which the state space $\mathbb{R}^2$ is partitioned into three regions $\mathcal{S}_1$, $\mathcal{S}_2$, and $\mathcal{S}_3$ by two parallel lines $H_{1,2}(x)=\begin{bmatrix}1&0\end{bmatrix}x$ and $H_{2,3}(x)=\begin{bmatrix}1&0\end{bmatrix}x-2$, see Fig. \ref{fig:numerical.example}. Note that Assumption~\ref{assum:lierequirement} is satisfied for  $\mathcal{C} = \mathbb{R}^2$.

For $Q=I_2$, one obtains:
\begin{align*}
    \mu_{Q,2}\left(\frac{\partial f_1}{\partial x}(x)\right)=
    \mu_2\left(\begin{bmatrix}-1&1\\0&-1\end{bmatrix}\right)=-0.50<0, \forall x \in \mathcal{S}_1\\
    \mu_{Q,2}\left(\frac{\partial f_2}{\partial x}(x)\right)=
    \mu_2\left(\begin{bmatrix}-2&1\\0&-1\end{bmatrix}\right)=-0.80<0,\forall x \in \mathcal{S}_2\\
    \mu_{Q,2}\left(\frac{\partial f_3}{\partial x}(x)\right)=
    \mu_2\left(\begin{bmatrix}-3&1\\0&-1\end{bmatrix}\right)=-0.88<0, \forall x \in \mathcal{S}_3
\end{align*}
i.e., the condition \eqref{eq:contractivePWScondition1} is fulfilled with $c=0.5$, and  \eqref{eq:contractivePWScondition2} is also satisfied:
\begin{align*}
    &\mu_{Q,2}((f_2(x)-f_1(x))\begin{bmatrix}1&0\end{bmatrix})=\mu_{2}\left(\begin{bmatrix}
        -x_1-2&0\\0&0\end{bmatrix}\right)\\
        &\qquad  \qquad =\max(-2,0)=0,\ \forall x\in \{H_{1,2}(x)= 0\},\\
    &\mu_{Q,2}((f_3(x)-f_2(x))\begin{bmatrix}1&0\end{bmatrix})=\mu_2\left(\begin{bmatrix}
        -x_1-3&0\\0&0\end{bmatrix}\right)\\
        &\qquad  \qquad  =\max(-3,0)=0,\ \forall x\in \{H_{2,3}(x)= 0\}.
\end{align*}
Therefore, according to Theorem \ref{theorem:originPWScontractive}, the Filippov solution of  \eqref{eq:generalPWSnew} is contracting in $\mathcal{C}$. Moreover, since the point $\hat{x}=\begin{bmatrix} 0.5& 0\end{bmatrix}^T$ is the only equilibrium point (and thus constitutes a solution) of   \eqref{eq:generalPWSnew}, and since there is no Filippov  equilibrium point on either switching manifold, $\hat{x}$ thus must be stable and attract all  Filippov solutions. This fact is confirmed in  Fig. \ref{fig:numerical.example} by showing that the crossing and sliding motions do not affect  the convergence towards the equilibrium $\hat{x}$.

\begin{figure}[t!]
  \begin{center}	
  \includegraphics[width = 0.45\textwidth]{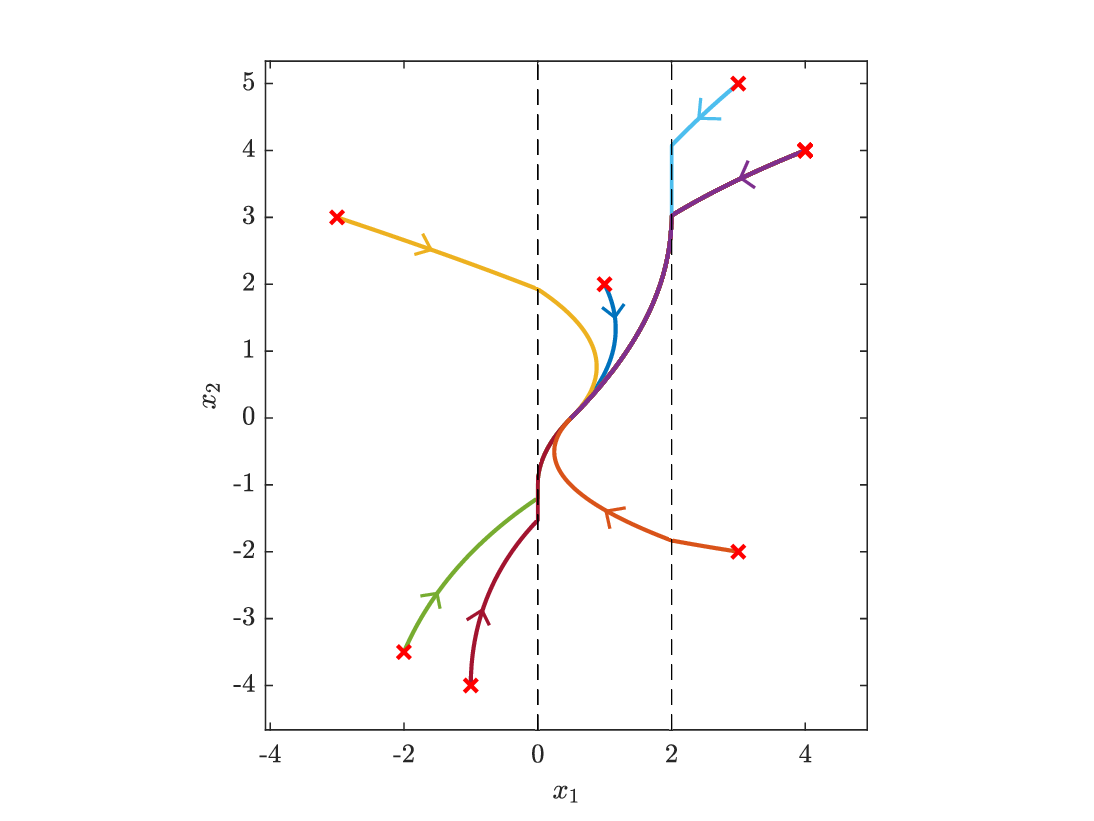}
   \caption{Evolution of the Filippov solution of Example 1 from different initial states with both sliding and crossing motions.}
    \label{fig:numerical.example}
      \end{center}
\end{figure}

\subsection{Example 2}
Consider now the PWS system of type \eqref{eq:generalPWSnew} with $N=4$ given by:
\begin{align*}
    f_1(x)&=\begin{bmatrix}-8&-2\\4&-4\end{bmatrix}x+\begin{bmatrix}6\\-1.5\end{bmatrix},\\ 
    f_2(x)&=\begin{bmatrix}-2&-4\\3&-4\end{bmatrix}x+\begin{bmatrix}4\\-1.5\end{bmatrix},\\
    f_3(x)&=\begin{bmatrix}-2&-2\\4&-10\end{bmatrix}x+\begin{bmatrix}4\\-1.3\end{bmatrix},\\
    f_4(x)&=\begin{bmatrix}-8&-4\\3&-10\end{bmatrix}x+\begin{bmatrix}6\\-1.3\end{bmatrix}.
\end{align*}
The state space $\mathbb{R}^2$ is partitioned into four domains $\mathcal{S}_1,\mathcal{S}_2,\mathcal{S}_3$ and $\mathcal{S}_4$ by two intersecting lines $H_1(x)=\begin{bmatrix}0&1\end{bmatrix}x$ and $H_2(x)=\begin{bmatrix}1&0\end{bmatrix}x$, see Fig.\eqref{fig:numerical.example2}.
With $Q=I_2$, it holds that:
\begin{align*}
    \mu_{Q}\left(\frac{\partial f_1}{\partial x}(x)\right)&=\mu_2\left(\begin{bmatrix}-8&-2\\4&-4
    \end{bmatrix}\right)\approx-3.76<0,\ &\forall x\in\mathcal{S}_1\\
     \mu_{Q}\left(\frac{\partial f_2}{\partial x}(x)\right)&=\mu_2\left(\begin{bmatrix}-2&-4\\3&-4
    \end{bmatrix}\right)\approx-1.88<0,\ &\forall x\in\mathcal{S}_2\\
    \mu_{Q}\left(\frac{\partial f_3}{\partial x}(x)\right)&=\mu_2\left(\begin{bmatrix}-2&-2\\4&-10
    \end{bmatrix}\right)\approx-1.87<0,\ &\forall x\in\mathcal{S}_3\\
    \mu_{Q}\left(\frac{\partial f_4}{\partial x}(x)\right)&=\mu_2\left(\begin{bmatrix}-8&-4\\3&-10
    \end{bmatrix}\right)\approx-7.88<0,\ &\forall x\in\mathcal{S}_4
\end{align*}
i.e.,  \eqref{eq:contractivePWScondition14region1} is satisfied with $c=1.87$. The condition \eqref{eq:contractivePWScondition14region4} is  satisfied in the  intersection point $\tilde{x}=\begin{bmatrix}0&0\end{bmatrix}^T$ since:
\begin{align*}
    f_2(x)+f_4(x)-f_1(x)-f_3(x)=\begin{bmatrix}0&0\end{bmatrix}^T,\ for\ x=\tilde{x}.
\end{align*}
Furthermore, the conditions \eqref{eq:contractivePWScondition14region2}-\eqref{eq:contractivePWScondition14region34} hold since:
\begin{align*}
 &   \mu_Q((f_1(x)+f_2(x)-f_3(x)-f_4(x))\nabla H_1(x))\\
  &  =\mu_2\left(\begin{bmatrix}0&0\\0&12x_2-4\end{bmatrix}\right)=0,\ \forall x\in\Sigma_1\\
  &  \mu_Q((f_2(x)+f_3(x)-f_1(x)-f_4(x))\nabla H_2(x))\\
  & = \mu_2\left(\begin{bmatrix}12x_1-4&0\\0&0\end{bmatrix}\right)=0,\ \forall x\in\Sigma_2\\
  &  \mu_Q((f_2(x)+f_4(x)-f_1(x)-f_3(x))\nabla H_1(x))\\
  & = \mu_2\left(\begin{bmatrix}0&-4x_2\\0&-2x_1 \end{bmatrix}\right)=0,\ \forall x\in\Sigma_1\cap\{H_2(x)>0\}\\
  &  \mu_Q((f_1(x)+f_3(x)-f_2(x)-f_4(x))\nabla H_1(x))\\
  & = \mu_2\left(\begin{bmatrix}0&4x_2\\0&2x_1\end{bmatrix}\right)=0,\ \forall x\in\Sigma_1\cap\{H_2(x)<0\}\\
  &  \mu_Q((f_2(x)+f_4(x)-f_1(x)-f_3(x))\nabla H_2(x))\\
  & = \mu_2\left(\begin{bmatrix}-4x_2&0\\-2x_1&0\end{bmatrix}\right)=0,\ \forall x\in\Sigma_2\cap\{H_2(x)>0\}\\
  &  \mu_Q((f_1(x)+f_3(x)-f_2(x)-f_4(x))\nabla H_2(x))\\
  & = \mu_2\left(\begin{bmatrix}4x_2&0\\2x_1&0\end{bmatrix}\right)=0,\ \forall x\in\Sigma_2\cap\{H_2(x)<0\}.
\end{align*}
Therefore, according to Theorem \ref{theorem:originPWScontractiveint} the Filippov solution of \eqref{eq:generalPWSnew} is contracting, and as $\hat{x}=\begin{bmatrix}1.1&0.45\end{bmatrix}^T$ is the only equilibrium point (including Filippov equilibrium points on the switching manifolds), it must again be stable and attract all other Filippov solutions, see Fig.\ref{fig:numerical.example2}.

\begin{figure}[t!]
  \begin{center}	
  \includegraphics[width = 0.45\textwidth]{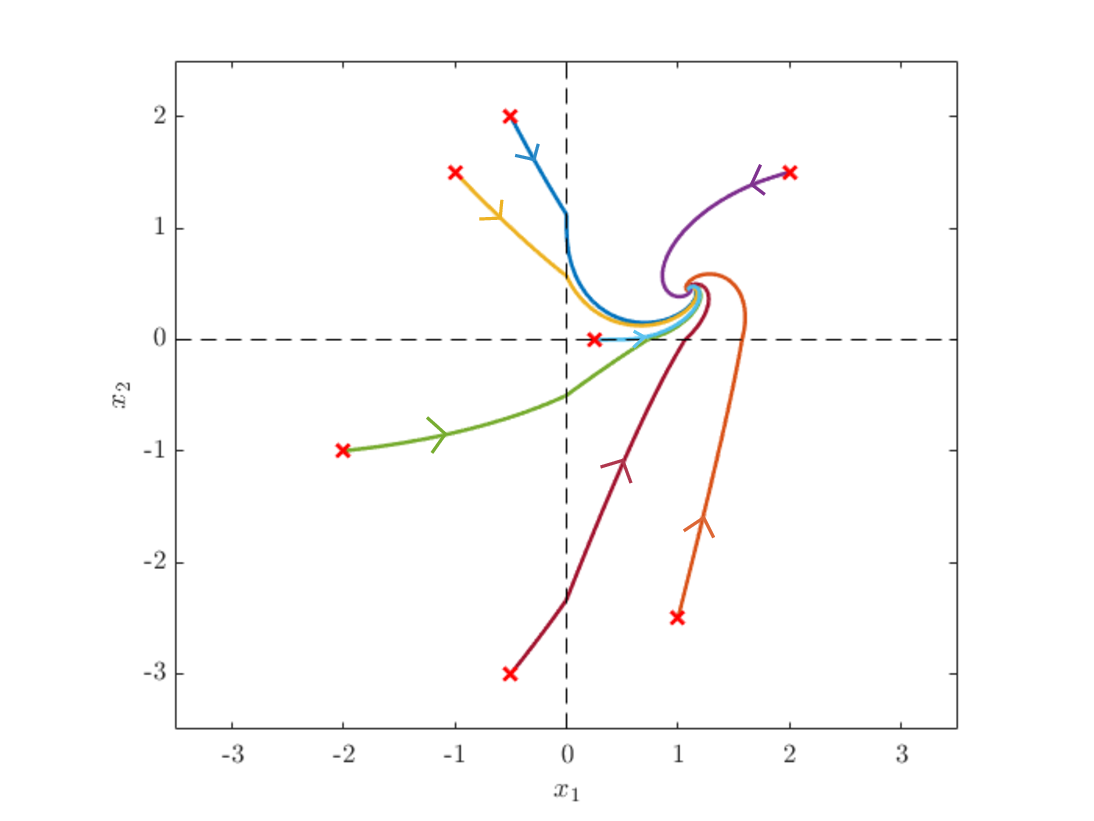}
   \caption{Evolution of the Filippov solution of Example 2 from different initial states with both sliding and crossing motions.}
    \label{fig:numerical.example2}
      \end{center}
\end{figure}

\section{Conclusion}
\label{sec:conclusion}
  In this paper, conditions to ensure contracting Filippov solutions for multi-modal PWS systems are provided. Compared to the contraction analysis of classical smooth dynamics, the sliding motion on multiple switching manifolds is taken into account in addition. Along this line,  the regularization-based approach, originally developed for the analysis of the sliding motion on single switching manifolds, is extended to a set of parallel, non-intersecting manifolds, and subsequently to intersecting manifolds. For the latter case, it is crucial to ensure that the intersections of switching manifolds do not belong to the sliding region, in order to avoid non-unique Filippov solutions. Future work aims at developing asymptotic tracking controllers for PWS systems by enforcing that the  controlled system is contractive, and that the reference trajectory belongs to the solutions of the controlled system.

\bibliography{IFACWC_Bibliography}             

\end{document}